\documentclass[11pt, a4paper]{article}
\usepackage{fullpage}

\usepackage{belov_paper}

\title{A Direct Reduction from the Polynomial to the Adversary Method}
\author{
  Aleksandrs Belovs\thanks{Faculty of Computing, University of Latvia} 
}
\date{}

\newcommand{\Pis}[1]{\Pi_{\le #1}}
\newcommand{\stDelta}[1]{\stackrel{\Delta_{#1}}{\longmapsto}}

\renewcommand{\*}{^{\phantom{*}}}

\begin{document}
\maketitle

\begin{abstract}
The polynomial and the adversary methods are the two main tools for proving lower bounds on query complexity of quantum algorithms.
Both methods have found a large number of applications, some problems more suitable for one method, some for the other.

It is known though that the adversary method, in its general negative-weighted version, is tight for bounded-error quantum algorithms, whereas the polynomial method is not.
By the tightness of the former, for any polynomial lower bound, there ought to exist a corresponding adversary lower bound.  However, direct reduction was not known.

In this paper, we give a simple and direct reduction from the polynomial method (in the form of a dual polynomial) to the adversary method.
This shows that any lower bound in the form of a dual polynomial is actually an adversary lower bound of a specific form.
\end{abstract}

\section{Introduction}
Proving lower bounds on quantum query complexity is a task that has attained significant attention.
The reason is that it is essentially the only known way to prove limitations on the power of quantum algorithms.
For instance, Bennett, Bernstein, Brassard, and Vazirani~\cite{bennett:strengths} proved a quantum query lower bound for the OR function using what later became known as the hybrid method.
This demonstrates that there is no way to attain a better than Grover's~\cite{grover:search} quadratic speed-up for an NP-search problem if we treat the latter as a black-box (an oracle).
Powerful tools for proving quantum query lower bounds have been developed consequently: the polynomial method, and the adversary method, both in its original (positive-weighted) and improved (negative-weighted) formulations.

The polynomial method is due to Beals, Buhrman, Cleve, Mosca, and de Wolf~\cite{beals:pol}, and it was inspired by a similar method used by Nisan and Szegedy~\cite{nisan:bs, nisan:pol} to prove lower bounds on randomized query complexity.
The method builds on the following observation: if $\cA$ is a $T$-query quantum algorithm, then its acceptance probability on input $x$ can be expressed as a degree-$2T$ multivariate polynomial in the input variables $x_i$.
Beals \etal~\cite{beals:pol} used this method to re-prove the lower bound for the OR function from~\cite{bennett:strengths}, and establish other results like a tight lower bound for all total symmetric Boolean functions.
A landmark result obtained by this method is the lower bound for the collision problem by Aaronson and Shi~\cite{shi:collisionLower}.
Similarly as Bennett \etal's result~\cite{bennett:strengths}, it shows that a black-box approach to finding a collision in a hash function by a quantum computer is doomed as well.
This method has been popular ever after.

The original adversary method is due to Ambainis~\cite{ambainis:adv}, and it is an improvement on the aforementioned hybrid method.
The bound was strengthened by Ambainis himself~\cite{ambainis:polVsQCC} and Zhang~\cite{zhang:advPower} shortly afterwards.
One of the appealing features of this method is its convenient combinatorial formulation, which resulted in a number of applications~\cite{barnum:readOnceLower, Durr:quantumGraph, buhrman:productVerification, dorn:algebraicProperties}.
However, the original formulation of the adversary bound was subject to several important limitations~\cite{zhang:advPower}.

Partly in order to overcome these limitations, H\o yer, Lee, and \v Spalek generalised the adversary bound in~\cite{hoyer:advNegative}.
Departing from the semidefinite formulation of the original adversary bound by Barnum, Saks, and Szegedy~\cite{barnum:advSpectral}, H\o yer \etal showed that the same expression still yields a lower bound if one replaces non-negative entries by arbitrary real numbers.  
This \emph{negative-weighted} formulation of the bound is strictly more powerful than the positive-weighted one, but it lacks the combinatorial convenience of the latter.
The bound turned out to be useful for composed functions~\cite{hoyer:advNegative} and sum-problems~\cite{belovs:kSumLower, belovs:onThePower}.
In a series of papers~\cite{reichardt:formulae, reichardt:spanPrograms, reichardt:advTight}, Reichardt \etal surprisingly proved that the negative-weighted version of the bound is tight for bounded-error algorithms!

The polynomial method, on the other hand, is known to be non-tight.
Ambainis~\cite{ambainis:polVsQCC} constructed a first super-linear separation between the two for total Boolean functions.  This was later improved to an almost quartic separation by Aaronson, Ben-David, and Kothari~\cite{aaronson:cheatSheets}, which is essentially tight~\cite{aaronson:Huang}.
For partial functions, the separations can be even more impressive~\cite{belovs:polynomialSeparations}.

The history of relationship between the adversary and the polynomial methods is rather interesting.
For instance, the AND of ORs function allows for a very simple adversary lower bound~\cite{ambainis:adv}, but its polynomial lower bound is more complicated and was
only obtained more than a decade later.
It was achieved independently by Sherstov~\cite{sherstov:andortree}, and Bun and Thaler~\cite{ bun:andortree} using the technique of dual polynomials~\cite{sherstov:patternMethod}.
The latter is the dual of an approximating polynomial is the sense of linear programming.
Therefore, by strong duality, their optimal values are exactly equal, and every polynomial lower bound can, in principle, be stated as a dual polynomial.
The technique of dual polynomials has been used by Bun, Kothari, and Thaler~\cite{bun:polynomialStrikesBack} to prove strong lower bounds for a number of problems like $k$-distinctness, image size testing, and surjectivity.
The first of them was later improved in~\cite{mande:kDistinctness}.
Similarly strong adversary lower bounds for these problems are not known.

Since the adversary method is tight, for every polynomial lower bound, there ought to exist a similarly good adversary lower bound.
However, a direct reduction was not known.
In this paper, we prove a simple direct reduction, giving a mechanical way of converting every dual polynomial into an adversary lower bound of a specific form.
We hope that this connection will give a better understanding of both techniques, and should enable their combined use, which could result in better lower bounds.
Contrary to the majority of papers dealing with the general adversary method, all proofs in this paper are fairly elementary.

A related result is a direct reduction from the polynomial method to \emph{multiplicative} adversary by Magnin and Roland~\cite{magnin:relationLowerBounds}, while we give a reduction to a more widely-used \emph{additive} adversary.
We also note that our construction has similarities to a recent powerful lower bound technique by Zhandry~\cite{zhandry:record, liu:multiCollisions}.
It would be interesting to understand the connection between the two better.
\medskip

The following result is the cornerstone of our reduction. 

\begin{thm}
\label{thm:identically}
Let $X$ and $Y$ be sets of inputs, and $\mu$ and $\nu$ be probability distributions on $X$ and $Y$, respectively.
Assume that, for any assignment $\alpha$ of size $\le 2m$, we have 
\begin{equation}
\label{eqn:identically}
\Pr_{x\gets \mu} [x\sim \alpha] = \Pr_{y\gets \nu} [y\sim \alpha].
\end{equation}
Then, the quantum query complexity of distinguishing $X$ and $Y$ is $\Omega(m)$.
\end{thm}

The result itself is actually known.
To represent this, we will give two proofs in this paper.
The first one in \rf{sec:prelim} uses the method of dual polynomials and it is purely for illustrative purposes.
The second proof is the main technical contribution of this paper, and it is done using the adversary method.
Let us give an short outline here.
The proof uses the following collection of vectors:
\[
v^X_\alpha = \sum_{x\in X: x\sim\alpha} \sqrt{\mu_x} \ket|x>
\qqand
v^Y_\alpha = \sum_{y\in Y: y\sim\alpha} \sqrt{\nu_y} \ket|y>,
\]
where $\alpha$ is an assignment of the input variables.
By the indistinguishability, for every $k\le m$, there exists a linear isometry $W_{\le k}$ that maps each $v^X_\alpha$ with $|\alpha|\le k$ into $v^Y_\alpha$ and is zero on the orthogonal complement of these vectors.
Informally, $W_{\le k}$ gives a binding between $X$ and $Y$ for an algorithm that has made $k$ queries.
The adversary matrix is $\Gamma = \sum_{k=0}^{m-1} W_{\le k}$.
It is easy to see it has norm $m$, and we prove that $\|\Gamma\circ\Delta_j\|\le 1$ for all $j$.
This proof is contained in Sections~\ref{sec:delta} and~\ref{sec:2sets}.
In \rf{sec:delta}, we only consider the space $\bR^X$, and in \rf{sec:2sets}, we substitute $\bR^X$ with $\bR^Y$ using indistinguishability.

In \rf{sec:polynomial}, we show how to use this result to transform a dual polynomial into an adversary bound.
The idea is that a dual polynomial gives probability distributions $\mu$ and $\nu$ on two sets $\tX$ and $\tY$ that are ``close'' to $X$ and $Y$ and that satisfy the promise of \rf{thm:identically}.
We first prove the lower bound in the form of the adversary for distinguishing two probability distributions from~\cite{belovs:merkle}, as we think it is conceptually closer to the dual polynomial.
Obtaining a standard worst-case adversary bound is also easy.
It is just the restriction $\tGamma \elem[X,Y]$ of the matrix $\tGamma$ obtained in the second proof of \rf{thm:identically} for $\tX$ and $\tY$.

\section{Preliminaries}
\label{sec:prelim}
For a positive integer $m$, let $[m]$ denote the set $\{1,2,...,m\}$.
For a predicate $P$, we write $1_P$ to denote the indicator variable that is 1 is $P$ is true, and 0 otherwise.

We consider partial functions $f\colon D \to \{0,1\}$ with $D\subseteq [q]^n$.
We denote $X = f^{-1}(1)$ and $Y = f^{-1}(0)$.
Thus, the function $f$ distinguishes $X$ and $Y$.
An element $x=(x_1,x_2,\dots,x_n)\in [q]^n$, is called an \emph{input}, the set $[q]$ is called the \emph{input alphabet}, and $x_j\in[q]$ are individual input symbols.

A \emph{measure} on a finite set $X$ is a function $\mu$ from $X$ to the set of non-negative real numbers.  We denote the value of $\mu$ on $x\in X$ by $\mu_x$.
The measure is a \emph{probability distribution} if $\sum_{x\in X}\mu_x = 1$.
We use $x\gets\mu$ to denote that $x$ is sampled from the probability distribution $\mu$.

An \emph{assignment} is a function $\alpha\colon S\to[q]$ defined on a subset $S$ of the set of indices $[n]$.  We write $x\sim\alpha$ if $x\in X$ \emph{agrees} with the assignment $\alpha$, that is, $x_j = \alpha(j)$ for all $j\in S$.
The weight $|\alpha|$ of the assignment is the size of $S$.
It is possible to have an empty assignment $\emptyset$ of zero weight, in which case, every input string agrees to it.

\paragraph{Linear Algebra}
An $X\times Y$ matrix is a matrix with rows labelled by the elements of $X$ and columns by the elements of $Y$.  The element of an $X\times Y$ matrix $A$ at the intersection of the $x$-th row and the $y$-th column is denoted by $A\elem[x,y]$.
For $X'\subseteq X$ and $Y'\subseteq Y$, the matrix $A\elem[X', Y']$ is the restriction of $A$ to the rows in $X'$ and the columns in $Y'$.
We identify a subspace and the corresponding orthogonal projector, which we usually denote by $\Pi$ with additional decorations.
An \emph{isometry} is a linear operator that preserves inner product.
We need the following well-known result:
\begin{lem}
\label{lem:isometry}
Assume $\cH$ and $\cK$ are two inner-product spaces.
Let $(v_i)_{i\in A} \subseteq \cH$ and $(w_i)_{i\in A}\subseteq \cK$ be two collections of vectors indexed by the same index set $A$.
Assume $\ip<v_i,v_j> = \ip<w_i,w_j>$ for all $i, j\in A$.  Then, there exists an isometry $T\colon \spn_i v_i \to \spn_i w_i$ such that $Tv_i = w_i$ for all $i$.
\end{lem}


\paragraph{Adversary Bound}
We use two different flavours of the negative-weighted adversary bound.
Here we give the canonical version from~\cite{hoyer:advNegative} and later we state the distributional version from~\cite{belovs:merkle}.

Assume we want to distinguish two sets of inputs $X, Y\subseteq [q]^n$ as above.
Let $\Gamma$ be a real $X\times Y$ matrix.
For $j\in[n]$, denote by $\Gamma\circ\Delta_j$ the matrix of the same dimensions as $\Gamma$ whose $(x,y)$-th entry is given by $\Gamma\elem[x,y]\cdot 1_{x_j\ne y_j}$.
In other words, the entries with $x_j = y_j$ are being erased (replaced by zeroes).

\begin{thm}[\cite{hoyer:advNegative}]
\label{thm:adv}
Assume that $\Gamma$ is an $X\times Y$ real matrix such that $\|\Gamma\circ \Delta_j\|\le 1$ for all $j\in [n]$.  Then, the (bounded-error) quantum query complexity of evaluating $f$ is $\Omega(\|\Gamma\|)$.
\end{thm}

The matrix $\Gamma$ from \rf{thm:adv} is called the \emph{adversary matrix}, and it is known that the bound of this theorem is tight~\cite{reichardt:advTight}.

As it can be guessed from the notation, the mapping $\Gamma\mapsto\Gamma\circ\Delta_j$
is usually expressed as an Hadamard product with a 01-matrix $\Delta_j$ of dimensions $X\times Y$.
However, we find it more convenient to think of it as a mapping.
In particular, we don't have to formally re-define the matrix $\Delta_j$ for matrices $\Gamma$ of different dimensions, and the matrix $\Delta_j$ almost never appears by itself.

The norm of the matrix $\Gamma\circ\Delta_j$ is not always easy to estimate.
The following trick from~\cite{lee:stateConversion} is of help here.  With some stretch of notation, we write 
$\Gamma\stDelta j B$ if $(\Gamma-B)\circ\Delta_j = 0$.
In other words, we are allowed to arbitrary change the $(x,y)$-entries of $\Gamma$ with $x_j=y_j$ in order to obtain $B$.
The idea is as follows:
\begin{prp}
\label{prp:trick}
For any $B$ with $\Gamma\stDelta j B$, we have $\|\Gamma\circ\Delta_j\| \le 2\|B\|$.
Moreover, if $f$ is a Boolean function, i.e., $D\subseteq \cube$, then $\|\Gamma\circ\Delta_j\| \le \|B\|$.
\end{prp}
Hence, we can bound $\|\Gamma\circ\Delta_j\|$ from above by estimating $\|B\|$, which is often easier.

\paragraph{Distributional Adversary}
We also use the version of the adversary bound for distinguishing two probability distributions.
This version is rather versatile as it allows the probability distributions to overlap and to have arbitrary acceptance probabilities.
 
\begin{thm}[\cite{belovs:merkle}]
\label{thm:advDistrib}
Let $D = [q]^n$.
Assume $\cA$ is a quantum algorithm that makes $T$ queries to the input string $x=(x_1,\dots,x_n)\in D$, and performs a measurement at the end with two outcomes 'accept' or 'reject'.
Let $\mu$ and $\nu$ be two probability distributions on $D$,
and denote by $s_\mu$ and $s_\nu$ the acceptance probability of $\cA$ when $x$ is sampled from $\mu$ and $\nu$, respectively.  Then,
\begin{equation}
\label{eqn:adv}
T = \Omega\s[ \min_{j\in[n]} \frac{\delta_\mu^*\Gamma\delta_\nu\* -  \tau(s_\mu,s_\nu) \|\Gamma\|} {\|\Gamma\circ\Delta_j \|} ],
\end{equation}
for any $D\times D$ matrix $\Gamma$ with real entries. 
Here,
\begin{equation}
\label{eqn:deltamu}
\delta_\mu\elem[x] = \sqrt{\mu_x}
\qqand
\delta_\nu\elem[y]=\sqrt{\nu_y}
\end{equation}
are unit vectors in $\bR^D$, and
\begin{equation}
\label{eqn:tau}
\tau(s_\mu,s_\nu) = \sqrt{\strut s_\mu s_\nu} + \sqrt{\strut (1-s_\mu)(1-s_\nu)} \le 1 - \frac{|s_\mu-s_\nu|^2}{8}.
\end{equation}
\end{thm}

\paragraph{Polynomials}
In the polynomial method, we have to assume that the function $f\colon D\to\bool$ is Boolean: $D\subseteq\cube$.
If this does not hold, one has to make the function Boolean.  A popular option is to introduce new variables $\widetilde{ x_{i,a}}$ with $i\in[n]$ and $a\in[q]$, defined by $\widetilde {x_{i,a}} = 1_{x_i=a}$.

For $S\subseteq [n]$, the corresponding \emph{character} is the function $\chi_S \colon \cube\to\{\pm1\}$ defined by $\chi_S(x) = \prod_i (-1)^{x_i}$.
The characters form a basis of the space of functions $\bR^{\cube}$.
Hence, every function $f\colon \cube\to\bR$ has a unique representation as a \emph{polynomial}: $f = \sum_{S\subseteq[n]} \alpha_S\chi_S$.  The size of the largest $S$ with non-zero $\alpha_S$ is called the \emph{degree} of $f$.  

A degree-$d$ \emph{polynomial} is any function $p\colon \cube\to\bR$ of degree at most $d$.
A degree-$d$ \emph{dual polynomial} is a function $\phi\colon\cube\to\bR$ satisfying 
\[
\sum_{x\in\cube} \abs|\phi(x)| = 1
\qqand
\sum_{x\in\cube} \phi(x)\chi_S(x) = 0\quad \text{for all $|S|\le d$}.
\]
It is easy to check that the second condition above is equivalent to the following one:
\begin{equation}
\label{eqn:dual2}
\sum_{x\sim \alpha} \phi(x) = 0\quad\text{for all assignments $\alpha$ with $|\alpha|\le d$.}
\end{equation}
Dual polynomials~\cite{sherstov:patternMethod} can be used to show inapproximability for real-valued total functions.
We may assume $d<n$, since every function can be represented by a degree-$d$ polynomial.
\begin{thm}
\label{thm:dualTotal}
Let $d<n$.  For any function $f\colon \cube\to\bR$, we have
\begin{equation}
\label{eqn:dualTotal}
\min_p \max_{x\in\cube } |f(x) - p(x)| = \max_{\phi} \sum_{x\in\cube} \phi(x)f(x),
\end{equation}
where $p$ ranges over all degree-$d$ polynomials and $\phi$ ranges over all degree-$d$ dual polynomials.
\end{thm}

Let us now turn to the case of partial functions $f\colon D\to\bool$ with $D\subseteq \cube$.  Again, we let $X = f^{-1}(1)$ and $Y = f^{-1}(0)$.

\begin{defn}
\label{defn:approx}
We say that a polynomial $p\colon\cube\to\bR$ $\eps$-approximates a partial function $f\colon D\to\bool$ with $D\subseteq \cube$ if
\begin{itemize}
\item for every $x\in D$, we have $|p(x) - f(x)|\le \eps$;
\item for every $x\in \cube$, we have $0\le p(x)\le 1$.
\end{itemize}
\end{defn}

The importance of this definition stems from the following result:
\begin{thm}[\cite{beals:pol}]
\label{thm:approx}
If a partial function $f\colon D\to\bool$ with $D\subseteq \cube$ can be evaluated by a $T$-query quantum algorithm with error at most $\eps$, then $f$ can be $\eps$-approximated by a polynomial of degree at most $2T$.
\end{thm}

The corresponding analogue of~\rf{thm:dualTotal} is slightly more involved.
\begin{thm}
\label{thm:dualPartial}
The best approximation distance $\eps$ as in \rf{defn:approx} of the function $f$ by a degree-$d$ polynomial is given by
\begin{equation}
\label{eqn:dualPartialObjective}
\max\sfig{\max_\phi \sC[ \sum_{x\in X} \phi^+(x) - \sum_{x\notin Y} \phi^-(x) ], 0},
\end{equation}
where the maximisation is over functions $\phi\colon \cube\to\bR$ satisfying
\begin{equation}
\label{eqn:dualPartial}
\sum_{x\in X} \phi^+(x) + \sum_{x\in Y} \phi^-(x) = 1
\qqand
\sum_{x\in \cube} \phi(x)\chi_S(x) = 0\quad \text{for all $|S|\le d$}.
\end{equation}
\end{thm}

Here $\phi^+(x) = \max\{0,\phi(x)\}$ and $\phi^-(x) = \max\{0,-\phi(x)\}$ are the positive and the negative parts of $\phi$, respectively.

The proofs of Theorems~\ref{thm:dualTotal} and~\ref{thm:dualPartial} are based on linear programming duality.
They are given in \rf{app:LP} for completeness.

\begin{proof}[Proof of \rf{thm:identically} using Dual Polynomials]
We may assume the function $f$ is Boolean.
It suffices to show that it cannot be approximated a polynomial of degree less than $2m$.
Let
\[
\phi(x) = \begin{cases}
\mu_x/2,&\text{if $x\in X$;}\\
-\nu_x/2,&\text{if $x\in Y$;}\\
0,&\text{otherwise.}
\end{cases}
\]
This function satisfies~\rf{eqn:dualPartial} with $d=2m$.
Indeed, the first condition follows from $\mu$ and $\nu$ being probability distributions, and the second one follows from~\rf{eqn:dual2} since
\[
\sum_{x\sim\alpha} \phi(x)
= \frac12 \Pr_{x\gets \mu} [x\sim \alpha] - \frac12 \Pr_{y\gets \nu} [y\sim \alpha] =0
\]
by~\rf{eqn:identically}.
The value of~\rf{eqn:dualPartialObjective} is $1/2$, meaning it is impossible to get a better than trivial approximation.
\end{proof}

\section{
\texorpdfstring{$\Delta$-decomposition of $\bR^X$}{Delta-decomposition of RX}
}
\label{sec:delta}

Let $X\subseteq[q]^n$ be a set of inputs, and let $\mu$ be some measure on $X$.
The goal of this section is to develop a decomposition of the space $\bR^X$ convenient for the $\Delta_i$ operation and that takes into account the measure $\mu$.

\paragraph{Definition of subspaces.}
For each assignment $\alpha$, define the following vector in $\bR^X$:
\[
v_\alpha = \sum_{x\sim \alpha} \sqrt{\mu_x} \ket|x>.
\]
Based on these vectors, we define a number of subspaces.
First, for $k\in\{0,1,\dots,n\}$:
\[
\Pis k = \spn_{\alpha\colon |\alpha|=k} v_\alpha.
\]

\begin{clm}
\label{clm:Pi}
We have $\Pis{k-1}\subseteq \Pis{k}$ and $\Pis{n} = \bR^X$.
\end{clm}

\pfstart
Let $\alpha$ be an assignment of weight $k-1$, and $i$ be an element of $[n]$ outside the domain of $\alpha$.  Then,
\[
v_\alpha = \sum_{a\in[q]} v_{\alpha\cup\{i\mapsto a\}},
\]
proving the first claim.

For the second claim, note that an assignment $\alpha$ of weight $n$ defines an individual input.
\pfend

This gives an orthogonal decomposition of $\bR^X$ into subspaces
\[
\Pi_k = \Pi_{\le k} \cap \Pi_{\le k-1}^\perp = \Pi_{\le k} - \Pi_{\le k-1}.
\]

\paragraph{Example.}
A simple example is $X = [q]^n$ with the uniform distribution $\mu_x$.
Define two orthogonal projectors on $\bR^q$: $E_0 = J_q/q$ and $E_1 = I_q - E_0$, where $J_q$ is the all-1 matrix.  Then,
\[
\Pi_k = \sum_{s\in\cube: |s|=k} E_{s_1}\otimes E_{s_2}\otimes\cdots\otimes E_{s_n},
\]
where $|s|$ is the Hamming weight.
These operators are similar to the ones used in the construction of the adversary lower bound for element distinctness~\cite{belovs:adv-el-dist} and sum-problems~\cite{belovs:onThePower}.

\paragraph{Action of $\Delta_1$.}
Let us consider the action of $\Delta_1$, the remaining $\Delta_j$ being analogous.
For that, we define the following variant of the above subspaces:
\[
\Pis{k}' = \spn_{\alpha\colon |\alpha|=k,\, \text{$\alpha$ defined on 1}} v_\alpha.
\]
In particular, we again have $\Pis n' = \bR^X$.
However, this time $\Pis 0'$ is the empty subspace.

\begin{clm}
\label{clm:PiPrim}
We have the following:
\begin{itemize}
\item[(a)] $\Pis{k-1}'\subseteq \Pis{k}'$;
\item[(b)] $\Pis{k-1} \subseteq \Pis{k}' \subseteq \Pis k$;
\item[(c)] $\Delta_1 \circ \Pis k'=0$.
\end{itemize}
\end{clm}

\begin{proof}
The proof of (a) is analogous to the proof of \rf{clm:Pi}.

The second inclusion of (b) holds because $\Pis k'$ is a span of a subset of vectors of $\Pis k$.
To prove the first inclusion of (b), it suffices to show that an arbitrary $v_\alpha$ with $|\alpha|=k-1$ is contained in $\Pis k'$.
The proof of that is analogous to the proof of \rf{clm:Pi}.
However, this time we take $i=1$ if $\alpha$ is not defined on 1 (and an arbitrary $i$ as before, otherwise).

Now let us prove (c).
Note that $\Pi'_{\le k}$ can be written as a direct sum 
\[
\Pi'_{\le k} = \bigoplus_{b\in[q]} \Pi'_{\le k, b}\negbigskip
\]
of orthogonal projectors
\[
\Pi'_{\le k, b} = \spn_{\alpha\colon |\alpha|=k,\; \alpha(1) = b} v_\alpha.
\]
Each $\Pi'_{\le k, b}$ acts on the subspace spanned by $x\in X$ with $x_1=b$.
Hence, $\Delta_1\circ\Pi'_{\le k, b} = 0$.  By linearity, $\Delta_1\circ\Pi'_{\le k}=0$.
\end{proof}

\paragraph{Standard Form of Adversary}
As a warm-up for the next sections, we describe the following ``standard'' form of the ``adversary'' matrix on $\bR^X$:
\begin{equation}
\label{eqn:standard}
\sum_{k=0}^{m-1} \Pi_{\le k} = \sum_{k=0}^m (m-k)\Pi_k.
\end{equation}
Clearly, the norm of this matrix is $m$.  
The action of $\Delta_1$ is defined as
\begin{equation}
\label{eqn:standardGamma'}
\sum_{k=0}^{m-1} \Pi_{\le k} \stackrel{\Delta_1}{\longmapsto} \sum_{k=0}^{m-1} \sA[\Pi_{\le k} - \Pi'_{\le k}],
\end{equation}
where we use point (c) of \rf{clm:PiPrim}.

\begin{clm}
\label{clm:standard}
The norm of the operator on the right-hand side of~\rf{eqn:standardGamma'} is 1.
\end{clm}

\pfstart
The operator in question is a sum of projectors $\Pis k - \Pis k'$.
By point (b) of \rf{clm:PiPrim}, we know that $\Pis k$ is contained in $\Pis {k+1}'$.
Hence, these projectors are pairwise orthogonal, and the norm of the operator is 1.
\pfend

In the following section, we will mimic this construction but $X\times Y$-matrices.

\section{Second Proof of \rf{thm:identically}}
\label{sec:2sets}
In this section, we give a proof of \rf{thm:identically}, which is based on the adversary method.



We use $v^X_\alpha$ and $v^Y_\alpha$ to denote vectors defined on the input sets $X$ and $Y$:
\[
v^X_\alpha = \sum_{x\in X: x\sim\alpha} \sqrt{\mu_x} \ket|x>
\qqand
v^Y_\alpha = \sum_{y\in Y: y\sim\alpha} \sqrt{\nu_y} \ket|y>,
\]
and similarly for $\Pis k^X$, $\Pis k^Y$, $\Pis k^{\prime X}$ and $\Pis k^{\prime Y}$.

Let $\alpha$ and $\beta$ be assignments of weight at most $m$. 
Note that
\[
\ip<v^X_\alpha, v^X_\beta> 
= \Pr_{x\sim \mu} [x\sim \alpha \wedge x\sim\beta]
= \Pr_{y\sim \nu} [y\sim \alpha \wedge y\sim\beta]
= \ip<v^Y_\alpha, v^Y_\beta>.
\]
Indeed, either $\alpha$ and $\beta$ contradict each other, in which case the both sides of the above equality are zero, or they can be merged into one assignment of size at most $2m$, in which case~\rf{eqn:identically} applies.

Hence, by \rf{lem:isometry}, there exists a linear operator $W$ that maps $v^Y_\alpha$ into $v^X_\alpha$ for each $|\alpha|\le m$.
It is an isometry from $\Pi_{\le m}^Y$ onto $\Pi_{\le m}^X$, and it induces the following isometries:
\[
W_{\le k}\colon \Pi_{\le k}^Y \to \Pi_{\le k}^X
\qqand
W'_{\le k}\colon \Pi_{\le k}^{\prime Y} \to \Pi_{\le k}^{\prime X}
\]
for all $k\le m$.

\begin{clm}
\label{clm:WPrim}
We have $\Delta_1\circ W'_{\le k} = 0$.
\end{clm}

\pfstart
The proof is similar to the point (c) of \rf{clm:PiPrim}.
\pfend

We use the adversary matrix similar to~\rf{eqn:standard}, but this time this is an $X\times Y$ matrix
\begin{equation}
\label{eqn:identicalGamma}
\Gamma = \sum_{k=0}^{m-1} W_{\le k} = W\sD[ {\sum_{k=0}^{m-1} \Pi^Y_{\le k} }].
\end{equation}
Since $W$ is an isometry, the norm of $\Gamma$ is $m$.
Again, the action of $\Delta_1$ is defined by
\[
\Gamma \stackrel{\Delta_1}{\longmapsto} 
\sum_{k=0}^m \sA[ W_{\le k} - W'_{\le k}] = 
W\sD[ {\sum_{k=0}^m \sA[\Pi^Y_{\le k} - \Pi^{\prime Y}_{\le k}] }],
\]
and the norm of the latter matrix is 1 by \rf{clm:standard} and using that $W$ is an isometry again.
At this point, we can use \rf{thm:adv}.
\medskip

The preceding derivation also works in the case when $\mu$ and $\nu$ have overlapping supports, or not totally concentrated on $X$ and $Y$.
In this case, it could be easier to prove quantum indistinguishability of $\mu$ and $\nu$ directly using \rf{thm:advDistrib}.
For this theorem, the following result is useful:

\begin{prp}
\label{prp:nuGammamu}
If $\mu$ and $\nu$ are probability distributions, and $\Gamma$ is defined as in~\rf{eqn:identicalGamma}, then
\[
\|\Gamma\| = \delta_\mu^* \Gamma \delta_\nu\* = m,
\]
where $\delta_\mu$ and $\delta_\nu$ are defined in~\rf{eqn:deltamu}.
\end{prp}

\pfstart
Observing the right-hand side of~\rf{eqn:standard} and the definition of $\Gamma$, we see that the maximal singular value of $\Gamma$ is $m$, with the corresponding left and right singular vectors lying in $\Pi^X_0$ and $\Pi^Y_0$, respectively.
By definition, $\Pi^X_0 = \Pis 0^X$ projects onto $v^X_\emptyset = \delta_\mu$ and
$\Pi^Y_0 = \Pis 0^Y$ projects onto $v^Y_\emptyset = \delta_\nu$, where $\emptyset$ denotes the empty assignment.
\pfend

\section{Polynomial Lower Bounds}
\label{sec:polynomial}
In this section, we demonstrate a direct conversion of a polynomial lower bound into an adversary lower bound.  
We do so by taking a dual polynomial that witnesses degree at least $d$ and convert it into an adversary bound of value $\Omega(d)$.

For warm-up, we consider the case of total functions in \rf{sec:total}, and then the general case of partial functions in \rf{sec:partial}.
In both cases, we use the distributional version of the adversary bound, \rf{thm:advDistrib}, which we find conceptually more appropriate in this case.
However, it is not hard to reduce to the usual version of the bound, \rf{thm:adv}, as well, which we do in \rf{sec:usual}.

\subsection{Total Functions}
\label{sec:total}
We start with the case when $f\colon\cube\to\bool$ is a total Boolean function.
Assume it cannot be $1/3$-approximated by a polynomial of degree $d$.
In this case, we can use \rf{thm:dualTotal}.
Let $\phi$ be a degree-$d$ dual polynomial attaining the maximum in~\rf{eqn:dualTotal}.
Thus,
\begin{equation}
\label{eqn:Totaldual1}
\sum_{x\in\cube} \phi(x)f(x) \ge 1/3.
\end{equation}
Our goal is to prove an adversary lower bound of $\Omega(d)$.

Let us define
\begin{equation}
\label{eqn:tXandtY}
\tX = \{x\in \cube \mid \phi(x)\ge 0 \}
\qqand
\tY = \{y\in \cube \mid \phi(y)< 0 \},
\end{equation}
and two measures
\[
\mu\colon \tX\to \bR,\, x\mapsto 2\phi(x)
\qqand
\nu\colon \tY\to \bR,\, y\mapsto -2\phi(y).
\]
From~\rf{eqn:dual2} applied to empty $\alpha$, we get that $\sum_x \phi(x)=0$.  Also, $\sum_x |\phi(x)|=1$.  Hence,
\begin{equation}
\label{eqn:muandnuAreProbabilities}
\sum_{x\in\tX} \mu_x = \sum_{y\in\tY} \nu_y = 1,
\end{equation}
that is, both $\mu$ and $\nu$ are probability distributions.

Using~\rf{eqn:dual2} again, we get that for each assignment $\alpha$ of weight at most $d$, we have
\begin{equation}
\label{eqn:dualidentitical}
\Pr_{x\gets \mu} [x\sim \alpha] = \Pr_{y\gets \nu} [y\sim \alpha].
\end{equation}
Thus, by~\rf{thm:identically}, the quantum query complexity of distinguishing $\tX$ and $\tY$ is $\Omega(d)$.
This is a nice development, but we would really like to prove the same result for the sets $X = f^{-1}(1)$ and $Y=f^{-1}(0)$.
Luckily, by condition~\rf{eqn:Totaldual1}, these sets are sufficiently well correlated.


Take $\tGamma$ as in~\rf{eqn:identicalGamma} with sets $\tX$ and $\tY$ and distributions $\mu$ and $\nu$.
By the results of \rf{sec:2sets}, we get that
\begin{equation}
\label{eqn:normtGamma}
\|\tGamma\| = \delta_\mu^* \tGamma \delta_\nu\* = d/2,
\qqand
\|\tGamma\circ\Delta_j\|\le 1\quad \text{for all $j\in[n]$}.
\end{equation}

Let $\cA$ be a quantum algorithm that evaluates the function $f$.  
We apply \rf{thm:advDistrib} on this algorithm and distributions $\mu$ and $\nu$.
Note that~\rf{eqn:Totaldual1} is equivalent to
\[
\sum_{x\in f^{-1}(1)} \mu_x - \sum_{y\in f^{-1}(1)} \nu_y \ge 2/3.
\]
This is the difference between the ``ideal'' acceptance probabilities of $\cA$ on $\mu$ and $\nu$, i.e, in the hypothetical case when the algorithm never errs.
Assuming the actual error of the algorithm $\cA$ is at most 1/6, we get that
\[
s_\mu - s_\nu \ge 1/3
\]
in notations of \rf{thm:advDistrib}.
From~\rf{eqn:tau}, we get that $\tau(s_\mu, s_\nu) \le 1-\Omega(1)$.
Pluging this and~\rf{eqn:normtGamma} into~\rf{eqn:adv}, we get that the query complexity of $\cA$ is $\Omega(d)$.

\subsection{Partial Functions}
\label{sec:partial}
Now let us consider the case of partial functions, for which we have to use \rf{thm:dualPartial}.  Again, assume that $\phi$ is optimal.
Then we have from~\rf{eqn:dualPartialObjective}:
\begin{equation}
\label{eqn:partialDual}
\sum_{x\in X} \phi^+(x) - \sum_{x\notin Y} \phi^-(x)  \ge 1/3.
\end{equation}
The sets $\tX$ and $\tY$ are still defined as in~\rf{eqn:tXandtY}.
In order to define $\mu$ and $\nu$, we have to choose a different scaling factor.

By~\rf{eqn:dualPartial}, we still have that
$
\sum_{x\in\cube} \phi^+(x) = \sum_{x\in\cube} \phi^-(x).
$
Also
\[
\sum_{x\in\cube} \phi^-(x) = 
\sum_{x\in Y} \phi^-(x) + \sum_{x\notin Y} \phi^-(x) 
\le \sum_{x\in Y} \phi^-(x) + \sum_{x\in X} \phi^+(x) -1/3
= 2/3,
\]
where we used~\rf{eqn:partialDual} and the first condition from~\rf{eqn:dualPartial}.
Let us denote the left-hand side of the above inequality by $M$.
Then, we can define probability distributions
\[
\mu\colon \tX\to \bR,\, x\mapsto \phi(x)/M
\qqand
\nu\colon \tY\to \bR,\, y\mapsto -\phi(y)/M.
\]
So that~\rf{eqn:partialDual} becomes
\begin{equation}
\label{eqn:mu-nu}
\sum_{x\in X} \mu_x - \sum_{y\notin Y} \nu_y \ge 1/2.
\end{equation}

The equation~\rf{eqn:dualidentitical} still holds, and we use the same construction of $\tGamma$, which still satisfies~\rf{eqn:normtGamma}.

Let $\cA$ be an algorithm that evaluates $f$ with error $\eps$.
Denote by $p_x$ the acceptance probability of the algorithm on input $x\in\cube$.
So, we have $p_x \ge 1-\eps$ for $x\in X$,  $p_x \le \eps$ for $x\in Y$, and $0\le p_x\le 1$ for all $x$.
Thus,
\begin{align*}
s_\mu - s_\nu 
&=
\sum_{x\in X} \mu_x p_x + \sum_{x\notin X} \mu_x p_x
- \sum_{y\in Y} \nu_y p_y - \sum_{y\notin Y} \nu_y p_y\\
&\ge (1-\eps) \sum_{x\in X} \mu_x 
- \eps  \sum_{y\in Y} \nu_y - \sum_{y\notin Y} \nu_y
 \ge \sum_{x\in X} \mu_x - \sum_{y\notin Y} \nu_y - 2\eps \ge 1/2 - 2\eps\ge 1/4,
\end{align*}
assuming $\eps\le 1/8$.

In the same way as in \rf{sec:total}, \rf{thm:advDistrib} implies that the query complexity of $\cA$ is $\Omega(d)$.

\subsection{Usual Version of the Adversary}
\label{sec:usual}
Here we give the same proof using the usual formulation of the adversary bound, \rf{thm:adv}.
Assume we are in the general case of partial functions of \rf{sec:partial}.
We define
\[
\Gamma = \tGamma\elem[X,Y].
\]
As $\Gamma$ is a sub-matrix of $\tGamma$, we get $\|\Gamma\circ\Delta_j\|\le 1$ for all $j$ from~\rf{eqn:normtGamma}.
It suffices to show that $\|\Gamma\| = \Omega\sA[\|\tGamma\|]$.

We know that $\tGamma\delta_\nu = \|\tGamma\|\delta_\mu$ by \rf{prp:nuGammamu}.
This gives us
\[
\norm| \tGamma\elem[X,\cube]\;\delta_\nu | =
\normA|\tGamma|\cdot \normA|\delta_\mu\elem[X] |.
\]
On the other hand,
\[
\norm| \tGamma\elem[X,\cube]\;\delta_\nu |
\le 
\norm| \tGamma\elem[X,Y]\; \delta_\nu \elem[Y] | + \norm| \tGamma\elem[X,\overline Y]\; \delta_\nu \elem[\overline Y] |
\le \normA| \tGamma\elem[X,Y] | + \normA| \tGamma| \cdot \norm|\delta_\nu \elem[\overline Y] |,
\]
where $\overline Y = \cube\setminus Y$.  
Thus,
\[
\normA| \tGamma\elem[X,Y] | \ge 
\normA|\tGamma| \sB[{\normA|\delta_\mu\elem[X] | - \norm|\delta_\nu \elem[\overline Y] |}]
= \normA|\tGamma| \frac{\normA|\delta_\mu\elem[X] |^2 - \norm|\delta_\nu \elem[\overline Y] |^2}{\normA|\delta_\mu\elem[X] | + \norm|\delta_\nu \elem[\overline Y] |}.
\]
From~\rf{eqn:mu-nu}, we get that
\[
\normA|\delta_\mu\elem[X] |^2 - \norm|\delta_\nu \elem[\overline Y] |^2 = \sum_{x\in X} \mu_x - \sum_{y\notin Y} \nu_y \ge 1/2.
\]
Also, $\normA|\delta_\mu\elem[X] | + \norm|\delta_\nu \elem[\overline Y] |\le 2$, hence, we obtain
\[
\normA| \tGamma\elem[X,Y] | 
\ge \frac14 \normA|\tGamma|,
\]
as required.

\section*{Acknowledgements}
The author is thankful to Shalev Ben-David for the suggestion to apply this construction to the polynomial method.

This work has been supported by the ERDF project number 1.1.1.5/18/A/020 ``Quantum algorithms: from complexity theory to experiment.''

\bibliographystyle{habbrvM}
{
\small
\bibliography{belov}
}

\appendix

\section{Linear Programming for Dual Polynomials}
\label{app:LP}

\subsection{Proof of \rf{thm:dualTotal}}
The left-hand side of~\rf{eqn:dualTotal} is equal to the optimal value of the following linear optimisation problem:
\begin{subequations}
\begin{alignat}{3}
 {\mbox{\rm minimise }} &\quad& \eps   \notag\\ 
 {\mbox{\rm subject to }} && f(x) - \sum_S \alpha_S \chi_S(x) &\le \eps &\quad& \text{for all $x\in\cube$;} \label{eqn:11} \\
 && f(x) - \sum_S \alpha_S \chi_S(x) &\ge -\eps && \text{for all $x\in\cube$;}\label{eqn:12} \\
 && \alpha_S&\in \bR &&\text{for all $S\subseteq [n],\; |S|\le d$;}\notag\\ 
 &&\eps&\in\bR. \notag
\end{alignat}
\end{subequations}

Let us write the Lagrangian with the dual variables $a_x\ge 0$ for~\rf{eqn:11} and $b_x\ge 0$ for~\rf{eqn:12}:
\begin{equation}
\label{eqn:1Lagrangian}
\eps - \sum_x a_x\sC[ \eps - f(x) + \sum_S \alpha_S \chi_S(x) ] 
     - \sum_x b_x\sC[ \eps + f(x) - \sum_S \alpha_S \chi_S(x) ]
\end{equation}
Let us denote $\phi(x) = a_x - b_x$, so that we can rewrite the last expression as
\begin{equation}
\label{eqn:1Lagrangian2}
\sum_x \phi(x)f(x) + \eps\sC[ 1 - \sum_x a_x -\sum_x b_x ] - \sum_S \alpha_S \sC[\sum_x \phi(x)\chi_S(x)].
\end{equation}
In the dual optimisation problem, all of the brackets in~\rf{eqn:1Lagrangian2} must be zero.

We can turn any dual polynomial into a feasible solution to the dual~\rf{eqn:1Lagrangian2} by taking $a_x = \phi^+(x)$ and $b_x = \phi^-(x)$.

For the opposite direction, consider optimal primal and dual solutions, whose values are equal due to strong duality.
If $\eps>0$, then, by complementary slackness, at most one of $a_x$ and $b_x$ is non-zero for each $x$, therefore, $|\phi(x)| = a_x + b_x$.
Hence, $\phi$ is a dual polynomial satisfying $\sum_x \phi(x)f(x)=\eps$.
If $\eps=0$, we can take $\phi$ equal to the normalised parity function.

\subsection{Proof of~\rf{thm:dualPartial}}
In this case, we have the following linear programming problem:

\begin{subequations}
\begin{alignat}{3}
 {\mbox{\rm minimise }} &\quad& \eps   \notag\\ 
 {\mbox{\rm subject to }} 
 && \sum_S \alpha_S \chi_S(x) &\ge 1-\eps &\quad& \text{for all $x\in X$;} \label{eqn:21} \\
 && \sum_S \alpha_S \chi_S(x) &\le \eps &\quad& \text{for all $x\in Y$;} \label{eqn:22} \\
 && \sum_S \alpha_S \chi_S(x) &\ge 0 &\quad& \text{for all $x\notin X$;} \label{eqn:23} \\
 && \sum_S \alpha_S \chi_S(x) &\le 1 &\quad& \text{for all $x\notin Y$;} \label{eqn:24} \\
 && \alpha_S&\in \bR &&\text{for all $S\subseteq [n],\; |S|\le d$;}\notag\\ 
 &&\eps&\in\bR. \notag
\end{alignat}
\end{subequations}
Let us write the Lagrangian with the dual variables $a_x,b_x,c_x,d_x\ge 0$ for~\rf{eqn:21}---\rf{eqn:24}, respectively:
\begin{equation}
\label{eqn:2Lagrangian}
\begin{aligned}
\eps\quad &-& \sum_{x\in X} a_x\sC[\sum_S \alpha_S \chi_S(x) - 1 + \eps] \quad
     &-& \sum_{x\in Y} b_x\sC[\eps - \sum_S \alpha_S \chi_S(x) ] \\
     &-& \sum_{x\notin X} c_x\sC[\sum_S \alpha_S \chi_S(x) ] \quad
     &-& \sum_{x\notin Y} d_x\sC[1-\sum_S \alpha_S \chi_S(x) ]
\end{aligned}
\end{equation}
Let us define
\[
\phi(x) =
\begin{cases}
a_x - d_x & \text{if $x\in X$;}\\
c_x - b_x & \text{if $x\in Y$;}\\
c_x - d_x & \text{if $x\notin X\cup Y$.}\\
\end{cases}
\]
Then, we can rewrite~\rf{eqn:2Lagrangian} as
\begin{equation}
\label{eqn:LastDualityJedi}
\sum_{x\in X} a_x - \sum_{x\notin Y} d_x 
+ \eps \sC[1- \sum_{x\in X} a_x - \sum_{x\in Y} b_x]
- \sum_S \alpha_S \sC[\sum_{x\in \cube} \phi(x)\chi_S(x)].
\end{equation}
Again, in the dual optimisation problem, all the brackets in~\rf{eqn:LastDualityJedi} must be zero.

If $\phi$ satisfies~\rf{eqn:dualPartial}, then we can take $a_x = \phi^+(x)$ for $x\in X$, $b_x = \phi^-(x)$ for $x\in Y$, $c_x = \phi^+(x)$ for $x\notin X$, and $d_x = \phi^-(x)$ for $x\notin Y$, and get a feasible solution to the dual.

For the opposite direction, consider optimal primal and dual solutions, whose values are equal by strong duality.
We may assume $\eps>0$.
By complementary slackness, for each $x$, at most one of the dual variables is non-zero, except for the case when $\eps=1/2$, in which case both $a_x$ and $b_x$ can be non-zero.
Either way, we get $a_x = \phi^+(x)$ for $x\in X$, $b_x = \phi^-(x)$ for $x\in Y$, and $d_x = \phi^-(x)$ for $x\notin Y$.
Thus we obtain the required dual formulation of \rf{thm:dualPartial}.

Let us note that maximisation with 0 is required in~\rf{eqn:dualPartialObjective}.
For example, consider the case $d=n-1$, and $X$ and $Y$ are of size 1.
The function can be approximated by a polynomial of degree at most 1, thus $\eps=0$.
On the other hand, by the second condition of~\rf{eqn:dualPartial}, $\phi$ must be equal to a multiple of the parity function.
It is easy to see that $\sum_{x\in X} \phi^+(x) - \sum_{x\notin Y} \phi^-(x)$ is actually negative in this case.

\end{document}